\documentclass[aps,prl,showpacs,twocolumn,nofootinbib,10pt]{revtex4-1}
 \pdfoutput=1

\usepackage{amsmath}
\usepackage{amssymb}
\usepackage{mathrsfs}
\usepackage{amsthm}
\usepackage{bm}
\usepackage{url}
\usepackage[T1]{fontenc}
\usepackage{csquotes}
\MakeOuterQuote{"}

\newtheoremstyle{note}
  {\topsep/2}               
  {\topsep/2}               
  {}                      
  {\parindent}            
  {\itshape}              
  {.}                     
  {5pt plus 1pt minus 1pt}
  {}

\theoremstyle{note}
\newtheorem{theorem}{Theorem}
\newtheorem{lemma}{Lemma}

\theoremstyle{definition}

\theoremstyle{remark}

\def\vec#1{\bm{#1}} 

\newcommand{\tr}{\operatorname{tr}}

\newcommand{\diag}{\operatorname{diag}}

\newcommand{\rep}{\mathrel{\widehat{=}}}

\newcommand{\rmi}{\mathrm{i}}
\newcommand{\rme}{\mathrm{e}}

\newcommand{\rmT}{\mathrm{T}}

\newcommand{\rmC}{\mathrm{C}}
\newcommand{\rmU}{\mathrm{U}}

\newcommand{\be}{\begin{equation}}
\newcommand{\ee}{\end{equation}}
\newcommand{\ba}{\begin{align}}
\newcommand{\ea}{\end{align}}

\def\<{\langle}  
\def\>{\rangle}  

\def\eqref#1{\textup{(\ref{#1})}}  
\newcommand{\eref}[1]{Eq.~\textup{(\ref{#1})}}
\newcommand{\Eref}[1]{Equation~\textup{(\ref{#1})}}

\newcommand{\thref}[1]{Theorem~\ref{#1}}

\newcommand{\thsref}[1]{Theorems~\ref{#1}}
\newcommand{\Thsref}[1]{Theorems~\ref{#1}}

\newcommand{\lref}[1]{Lemma~\ref{#1}}

\newcommand{\lsref}[1]{Lemmas~\ref{#1}}

\newcommand{\cref}[1]{Conjecture~\ref{#1}}
\newcommand{\Cref}[1]{Conjecture~\ref{#1}}

\newcommand{\rcite}[1]{Ref.~\cite{#1}}
\newcommand{\rscite}[1]{Refs.~\cite{#1}}

\begin{document}
\title{Quasiprobability representations of quantum mechanics with minimal negativity}
\author{Huangjun Zhu}
\email{hzhu1@uni-koeln.de, zhuhuangjun@gmail.com}
\affiliation{Institute for Theoretical Physics, University of Cologne,  Cologne 50937, Germany}




\begin{abstract}
Quasiprobability representations, such as the Wigner function, play an important role in various research areas. The inevitable appearance of negativity in such representations is often regarded as a signature of nonclassicality, which has profound implications for quantum computation. However, little is known about the minimal  negativity that is necessary in general quasiprobability representations. Here we focus on  a natural class of quasiprobability representations that is distinguished by simplicity and economy. We  introduce  three measures of negativity concerning the representations of quantum states, unitary transformations, and quantum channels, respectively. Quite surprisingly, all three measures lead to the same 
 representations with minimal negativity, which are in one-to-one correspondence with the elusive symmetric informationally complete measurements. In addition,   most representations with minimal negativity are automatically covariant with respect to the Heisenberg-Weyl groups. Furthermore, our study  reveals an interesting tradeoff between negativity and symmetry in quasiprobability representations.

\end{abstract}

\date{\today}
\maketitle

Quasiprobability representations (QPRs) provide an intuitive  way of formulating  quantum theory in close analogy to classical theory. Archetypal  QPRs include  the Wigner function \cite{HillOSW84} and its discrete analogs \cite{Woot87,GibbHW04, Gros06, FerrE09, Ferr11, Zhu16P}, which are very useful  in many research areas, such as quantum optics, quantum tomography, and quantum computing. 
The inevitable appearance of negativity in such representations is often regarded as a signature of nonclassicality, which is  closely related to other nonclassical features \cite{KenfZ04, Spek08, FerrE08,FerrE09, Ferr11}. In particular, it was shown by Spekkens  that negativity and contextuality are equivalent notions of nonclassicality \cite{Spek08}. 
Moreover,  negativity has been recognized 
as a resource in quantum computation and has attracted increasing attention in the quantum information community \cite{Galv05,CormGGP06,VeitFGE12, VeitMGE14,HowaWVE14, DelfABR15,PashWB15}. To be specific, in the paradigm of magic state quantum computation \cite{BravK05}, negativity in the Wigner function is necessary for  computational speedup and is directly related to the efficiency of classical simulation \cite{VeitFGE12, VeitMGE14,PashWB15}.

To fully understand the distinction between quantum theory and classical  theory in terms of  negativity, one should not only consider the representation of quantum states, but also representations of measurements and  transformations. In addition, it is  crucial to consider the whole class of  QPRs instead of a particular choice. This advice is also helpful for understanding the power of quantum information processing as well as the connection between negativity and contextuality  \cite{Spek08, FerrE08,FerrE09}. However, little is known about the degree of negativity in any QPR other than the Wigner function and its discrete analogs.

In this paper we initiate a systematic study of the degree of negativity in general QPRs. Our focus is an important family of QPRs in finite dimensions that is distinguished by simplicity and economy. It  includes most examples proposed in the literature, such as many discrete Wigner functions \cite{Woot87,GibbHW04, FerrE09, Ferr11}. In these representations, information is encoded without redundancy; the analogy to classical probability theory is preserved as far as possible except for the appearance of negativity.

We then introduce three measures of negativity concerning representations of quantum states, unitary transformations, and quantum channels, respectively. Quite surprisingly, all of them lead to the same  representations with minimal negativity, which are in one-to-one correspondence with the  elusive symmetric informationally complete measurements (SICs for short) \cite{Zaun11, ReneBSC04, Appl05, ScotG10,  ApplFZ15G}.  The importance of SICs to foundational studies was first emphasized by Fuchs \cite{Fuch04}. Now SICs are  useful in many different contexts, especially  in the study of quantum Bayesianism \cite{Fuch10,ApplEF11,FuchS13}. Our results further corroborate the key role of SICs in understanding the nonclassicality inherent in quantum mechanics. 
We also  classify all representations with minimal negativity in dimensions 2 and 3 and explain the situation in higher dimensions. Most of these representations  are  covariant with respect to the Heisenberg-Weyl (HW) groups, which are intimately connected to another nonclassical feature, namely, uncertainty relations.
In addition, we  determine those unitary transformations with a nonnegative representation and elucidate  the distinction between  representations with minimal negativity and discrete Wigner functions. Our study  reveals an interesting tradeoff between negativity and symmetry, which has implications for both quantum foundations and quantum computation.

A QPR of quantum states is a linear map from quantum states to normalized real functions  on  a certain set, such as  phase spaces for  discrete Wigner functions \cite{Woot87,GibbHW04, FerrE09, Ferr11}. It is  conveniently expressed in terms of an \emph{operator frame}, which is a set of  operators that span the operator space. Here we are only concerned with frames composed of Hermitian operators. Any operator frame $\{F_j\}$ provides a QPR via the map $\rho\rightarrow \mu_j(\rho):=\tr (\rho F_j)$. Conversely, any QPR can be cast into this form \cite{FerrE08, FerrE09}. The normalization condition $\sum_j \mu_j(\rho)=1$ implies that $\sum_j F_j=1$. 
 Any operator frame in dimension $d$ has at least $d^2$ elements; those with $d^2$ elements are  minimal and form operator bases.

 To formulate a QPR of quantum mechanics, we also need to provide a suitable representation of measurements \cite{FerrE08,FerrE09}, characterized by positive-operator-valued measures (POVMs). Recall that a POVM $\{M_\xi\}$ is composed of a set of positive operators,  known as effects, that sum up to the identity operator 1. 
 To this end it is convenient to introduce the concept of dual frames. A dual frame  of  $\{F_j\}$ is a frame $\{Q_j\}$ that satisfies the  equation $
 \sum_j \tr(B F_j)\tr(Q_j C)=\tr(BC)$
 for any pair of operators $B,C$ acting on the given Hilbert space. If $Q_j$ are constrained to satisfy $\tr(Q_j)=1$, then 
 $\sum_\xi \nu_j(M_\xi)=1$ for $\nu_j(M_\xi)=\tr(Q_jM_\xi)$, so that $\nu_j(M_\xi)$ can be interpreted as conditional quasiprobabilities, when rewriting the Born rule as 
 \begin{equation}
\tr(\rho M_\xi)=\sum_j\mu_j(\rho)\nu_j(M_\xi).
 \end{equation} 
The dual frame is generally not unique, but it is unique for a minimal frame, in which case it is determined by the equation $\tr(F_jQ_k)=\delta_{jk}$.

 A  frame $\{F_j\}$ is self dual if there exists a dual frame $\{Q_j\}$ that is proportional to the original frame, that is, $Q_j=c F_j$ for some constant $c$. A minimal frame is self dual iff it is an orthonormal basis up to a scaling constant.
A QPR  based on a minimal self dual  frame  is \emph{normal}.
Now the constraints $\sum_j F_j=1$ and $\tr(Q_j)=1$ imply that $Q_j=dF_j$, $\tr(F_jF_k)=\delta_{jk}/d$, and $\tr(Q_jQ_k)=d\delta_{jk}$.
The representation is completely determined by the dual frame $\{Q_j\}$ via the equations $\mu_j(\rho)=\tr(\rho Q_j)/d$ and $\nu_j(M_\xi)=\tr(M_\xi Q_j)$.
When there is no danger of confusion, we shall identify $\{Q_j\}$ with the normal quasiprobability representation (NQPR).

 NQPRs include  most discrete Wigner functions proposed in  literature \cite{Woot87,GibbHW04,FerrE09, Ferr11}.
 They are particularly appealing because of a number of merits: there is no redundancy in the representation; the symmetry between quantum states and effects is preserved as far as possible; it is easy to convert between states, effects and their representations. In addition, unitary transformations can be represented by doubly quasistochastic matrices in close analogy to doubly stochastic matrices, as shown below.

 Given a unitary $U$ and NQPR $Q=\{Q_j\}$ in dimension~$d$, we have $\mu_j(U\rho U^\dag)=\sum_k U^Q_{jk} \mu_k(\rho)$, where  $U^Q$ is an orthogonal matrix with  $U^Q_{jk}=\tr(Q_j UQ_kU^\dag)/d$ satisfying $\sum_{j} U^Q_{jk}=\sum_{k} U^Q_{jk}=1$ since $\tr(Q_j)=1$ and $\sum_j Q_j=d$. In other words, $U^Q$ would be a doubly stochastic matrix if all its entries were nonnegative. In general, they are  doubly quasistochastic. The above representation can  be generalized to arbitrary quantum channels characterized by completely positive trace-preserving maps from dimension $d$ to dimension $d$. For example, the channel $\Lambda$ can be represented by the matrix $\Lambda^Q$ with $\Lambda^Q_{jk}=\tr[Q_j \Lambda(Q_k)]/d$. The equality $\sum_{j} \Lambda^Q_{jk}=1$ still holds for all $k$; the other equality $\sum_{k} \Lambda^Q_{jk}=1$  holds for all $j$ iff the channel is unital.

It is known  that the appearance of negative probabilities is inevitable in   any QPR of quantum mechanics \cite{Spek08, FerrE08,FerrE09}.  However, little is known about the amount of negativity that is necessary. In this paper we shall fill this gap. The \emph{negativity} of a quantum state $\rho$ with respect to a NQPR $\{Q_j\}$ is defined as
\begin{equation}
N(\rho)=d\max\{0,-\min_j \mu_j(\rho)\}=\max\{0,-\min_j \tr(\rho Q_j)\}.
\end{equation}
Alternative definitions  (such as the sum of negative entries \cite{VeitMGE14}) have also been studied 
for discrete Wigner functions.  These measures are  useful in a number of  contexts, especially in quantum computation \cite{DamH11,VeitFGE12,VeitMGE14}.
The negativity $N(\{Q_j\})$ of the NQPR $\{Q_j\}$  is defined as the maximum of $N(\rho)$ over all quantum states, that is,
\begin{equation}
N(\{Q_j\})= \Bigl|\min_j\lambda_{\min}(Q_j)\Bigr|,
\end{equation}
where $ \lambda_{\min}(Q_j)$ is the minimal eigenvalue of $Q_j$, which must be negative given the constraints $\tr(Q_j)=1$ and $\tr(Q_j^2)=d$. 
Note that orthonormal operator bases cannot consist of only positive operators \cite{ApplDF14,Spek08,FerrE08}.

The nonclassical features of a QPR are not only reflected in the negative probabilities but also in the negative entries in the representation of unitary transformations and quantum channels. The \emph{channel negativity} of a  channel $\Lambda$ with respect to a NQPR $\{Q_j\}$ is defined as
\begin{equation}
N_\rmC(\Lambda)=\max\bigl\{0,-d\min_{jk}\Lambda^Q_{jk}\bigr\}.
\end{equation}
When $\Lambda$ denotes conjugation by the unitary $U$, we get the \emph{unitary negativity}
\begin{equation}\label{eq:NegU0}
N_\rmU(U)=-d\min_{jk} U^Q_{jk}=d\Bigl|\min_{jk} U^Q_{jk}\Bigr|.
\end{equation}
Note that $\min_{jk} U^Q_{jk}\leq 0$ since $U^Q$ is orthogonal and $\sum_{j} U^Q_{jk}=\sum_{k} U^Q_{jk}=1$. 
The channel negativity $N_\rmC(\{Q_j\})$ of the  representation $\{Q_j\}$  is  the maximum of $N_\rmC(\Lambda)$ over all quantum channels; the unitary negativity $N_\rmU(\{Q_j\})$ is the maximum of $N_\rmU(U)$ over all unitaries. Let $\lambda_j$ be the vector of eigenvalues of $Q_j$; then we have
\begin{align}
N_\rmU(\{Q_j\})&= \Bigl|\min_{j,k}\lambda_j^\uparrow \cdot \lambda_k^\downarrow\Bigr|,\label{eq:NegU}\\
N_\rmC(\{Q_j\})&= \max_{j,k}\frac{\lambda_{j}^{\max}(\|\lambda_k\|_1-1)- \lambda_{j}^{\min}(\|\lambda_k\|_1+1)}{2}.  \label{eq:NegC}
\end{align}
Here $\lambda_j^\uparrow$ ($\lambda_j^\downarrow$) is the vector of eigenvalues of $Q_j$ arranged in nondecreasing (nonincreasing) order, $\lambda_{j}^{\max}=\max_r \lambda_{j,r}$,  $\lambda_{j}^{\min}=\min_r \lambda_{j,r}$, and $\|\lambda_k\|_1=\sum_{r=1}^d |\lambda_{k,r}|=\|Q_k\|_1$. \Eref{eq:NegU} follows from \eref{eq:NegU0}
and the tight inequality $\tr(Q_jUQ_kU^\dag)\geq\lambda_j^\uparrow \cdot \lambda_k^\downarrow$ (see page 341 of \rcite{MarsOA11book}). \Eref{eq:NegC} is derived in the supplementary material.

Before presenting our main results, it is instructive to  compute the negativity of typical NQPRs. The Wootters discrete Wigner function \cite{Woot87}  is determined by the basis of phase point operators $\{A_j\}$ (corresponding to $\{Q_j\}$ in the above discussion), where $j$ may represent multiple indices. In dimension 2,
\begin{equation}\label{eq:Woot87}
A_{j_1j_2}=\frac{1}{2}\bigl[1+(-1)^{j_1}\sigma_z+(-1)^{j_2}\sigma_x+(-1)^{j_1+j_2}\sigma_y\bigr]
\end{equation}
for $j_1,j_2=0,1$, so  the negativity is equal to $(\sqrt{3}-1)/2$.  Interestingly, the four states achieving the maximal negativity happen to be magic states \cite{BravK05}, which play a key role in quantum computation.  In each odd prime dimension $p$, the Wootters discrete Wigner function is uniquely determined by Clifford covariance \cite{Gros06,Zhu16P}; each phase point operator has two distinct eigenvalues $\pm1$ with multiplicities $(p\pm1)/2$. In general, the phase point operators are tensor products of those phase point operators in prime dimensions \cite{Woot87,Gros06,Zhu16P}. Let $n$  be the largest exponent such that $2^n$ divides the dimension~$d$. Then the negativity, unitary negativity, and channel negativity read
\begin{align}
N&=\begin{cases}
2^{1-n}(\sqrt{3}+1)^{n-2}&  d=2^n,\\
2^{-n}(\sqrt{3}+1)^{n}& d\neq2^n;
\end{cases}\\
N_\rmU&=d-2^{|n-1|};\\
N_{\rmC}&=\begin{cases}
2^{-n}(\sqrt{3}+1)^{n-1}(3^{(n+1)/2}-1)& d=2^n, \\
4^{-n}(\sqrt{3}+3)^{n}d&  d\neq2^n.
\end{cases}
\end{align}
These results reflect a significant distinction between discrete Wigner functions in even and odd dimensions. Incidentally, the Wigner function in the continuous scenario is bounded from below by $-1/(\pi\hbar)$ \cite{Roye77}.

To construct NQPRs with smaller negativity,   recall that a   SIC $\{\Pi_j\}$ in dimension $d$ is composed of $d^2$ rank-1 projectors with equal pairwise fidelity,   $\tr(\Pi_j\Pi_k)=(d\delta_{jk}+1)/(d+1)$ \cite{Zaun11,ReneBSC04, Appl05, ScotG10, ApplFZ15G}.
This definition implies the
normalization condition $\sum_j\Pi_j=d$ automatically. 
Given  a SIC $\{\Pi_j\}$, we can construct two  NQPRs $\{Q_j^+\}$ and $\{Q_j^-\}$,
\begin{align}
Q_j^\pm&=\mp\sqrt{d+1}\Pi_j+\frac{1}{d}(1\pm \sqrt{d+1}). \label{eq:SICframe} 
\end{align}
When $d=2$, the two NQPRs are equivalent and both of them are equivalent to the Wootters discrete Wigner function.  Here two NQPRs $\{Q_j\}$ and $\{Q_j'\}$ are equivalent if there is a unitary transformation $U$ such that $U Q_jU^\dag=Q'_{\sigma(j)}$ for some permutation $\sigma$.
 In general, the above representations are intimately connected to the representation of quantum states in terms of SIC probabilities, which have been emphasized in quantum Bayesianism \cite{Fuch04,Fuch10,ApplEF11, FuchS13} for their heuristic value in making the Born rule look similar to the classical Law of Total Probability \cite{Fuch14}.  In addition, the underlying operator bases are useful for the study of the Lie and Jordan algebraic structures underlying quantum theory \cite{ApplFZ15G}.  Here they are interesting because of their key role in understanding the minimal negativity in NQPRs, which as a bonus provides for the first time a quantitative argument for the Born rule representation sought in quantum Bayesianism.

The negativity, unitary negativity, and channel negativity  of the  NQPRs $\{Q_j^+\}$ and $\{Q_j^-\}$ read
\begin{align}
N^+&=\frac{(d-1)\sqrt{d+1}-1}{d}, \quad N^-=\frac{\sqrt{d+1}-1}{d},  \\
N_\rmU^\pm&=1,\quad
N_\rmC^{\pm}=\frac{d^2-2\pm (d-2)\sqrt{d+1}}{d}.
\end{align}
All three  measures $N^-, N_\rmU^-, N_\rmC^-$ (especially the first two) of $\{Q_j^-\}$ are smaller than the corresponding values for the Wootters discrete Wigner function when $d>2$. What is not so obvious  is that $N^-, N_\rmU^-, N_\rmC^-$ are  the lower bounds for  negativity, unitary negativity, and channel negativity.



\begin{theorem}\label{thm:NegBound}
 Any NQPR $\{Q_j\}$ in dimension $d$ satisfies $N^-\leq N(\{Q_j\})\leq N^+$.
The lower bound is saturated iff $\{Q_j\}$ has the form $\{Q_j^-\}$  with  $\{\Pi_j\}$ being a SIC. If  $\{Q_j\}$ is group covariant, then the upper bound is saturated iff $\{Q_j\}$ has the form $\{Q_j^+\}$ with  $\{\Pi_j\}$ being a SIC.
\end{theorem}
The representation $\{Q_j\}$ is group covariant if its symmetry group acts transitively on its frame elements. Here the symmetry group  is the group of all unitary transformations that leaves the set of frame elements  invariant.
NQPRs saturating the lower bound in \thref{thm:NegBound} are called \emph{perfect}. In dimension 2, the lower and upper bounds in the theorem coincide, so all NQPRs are  perfect; actually, all of them are equivalent to the Wootters discrete Wigner function since all SICs are  equivalent.

\begin{proof}
Let $\lambda_{j,1}, \lambda_{j,2}, \ldots,\lambda_{j,d}$ be the eigenvalues of $Q_j$  in nonincreasing order. The constraints $\tr(Q_j)=1$ and $\tr(Q_j^2)=d$ amount to the equalities $\sum_r\lambda_{j,r}=1$ and $\sum_r\lambda_{j,r}^2=d$. When $\lambda_{j,d}$ is minimized (maximized), the first (last) $d-1$ components of $\lambda_j$ must be equal, so that
 $N^-\leq -\lambda_{j,d}\leq N^+$.
The lower bound is saturated iff
\begin{equation}\label{eq:LowerSpectrum}
\lambda_{j,1}=\frac{(d-1)\sqrt{d+1}+1}{d}, \quad
\lambda_{j,2}=\cdots=\lambda_{j,d}=-N^-;
\end{equation}
the upper bound is saturated iff
\begin{equation}\label{eq:UpperSpectrum}
\lambda_{j,1}=\cdots=\lambda_{j,d-1}=\frac{\sqrt{d+1}+1}{d}, \quad
 \lambda_{j,d}=-N^+.
\end{equation}
It follows that $N^-\leq N(\{Q_j\})\leq N^+$.

If the lower bound is saturated, then \eref{eq:LowerSpectrum} is satisfied for all $j$. So $Q_j$ have the form $Q_j^-$ in \eref{eq:SICframe} with $\Pi_j$ rank-1 projectors. The requirement $\tr(Q_jQ_k)=d\delta_{jk}$ then implies that $\tr(\Pi_j\Pi_k)=(d\delta_{jk}+1)/(d+1)$; that is, $\{\Pi_j\}$ is a SIC.
To saturate the upper bound, it suffices to have one $Q_j$ possess the spectrum in \eref{eq:UpperSpectrum}; such operator bases can be constructed for any dimension~$d$. On the other hand, if  all  $Q_j$ have the same minimal eigenvalue, which holds if $\{Q_j\}$ is group covariant, then all $Q_j$  have the form  $Q_j^+$ in \eref{eq:SICframe}, with $\Pi_j$ forming a SIC.
\end{proof}

More is known about NQPRs with maximal negativity.
 \begin{theorem}\label{thm:NegBound2}
 	In any NQPR $\{Q_j\}$ in dimension $d$,  the number of states with negativity  $N^+$ is at most $d^2$; the bound is saturated  iff $\{Q_j\}$ has the form $\{Q_j^+\}$  with  $\{\Pi_j\}$ being a SIC.
 \end{theorem}
 Here $N^+$ is the maximal negativity over all quantum states and all NQPRs.
 The proof is  similar to the proof of \thref{thm:NegBound}. Note that if a quantum state $\rho$ satisfies $\tr(\rho Q_j)=-N^+$, then $Q_j$ has  the spectrum in \eref{eq:UpperSpectrum}, and $\rho$ is the unique eigenstate of $Q_j$ with the minimal eigenvalue.

\begin{theorem}\label{thm:NegU}
	The unitary negativity of any NQPR $\{Q_j\}$ in dimension $d$ satisfies $1\leq N_\rmU(\{Q_j\})\leq d- (2/d)$.
	The lower bound is saturated iff $\{Q_j\}$ has the form $\{Q_j^-\}$ or $\{Q_j^+\}$   with  $\{\Pi_j\}$ being a SIC.
\end{theorem}
Interestingly, the lower bound and the equality condition do not change if $ N_\rmU(\{Q_j\})$ is replaced by the minimum of $N_\rmC(\Lambda)$ over unital channels $\Lambda$. 
\begin{proof}
	The lower bound and the equality condition follows from Theorem~7 in \rcite{ApplFZ15G}. Let $\eta_j=\lambda_j-(1/d)$, then  the upper bound follows from the  equation
	\begin{equation}
	\lambda_j^\uparrow \cdot \lambda_k^\downarrow=\eta_j^\uparrow \cdot \eta_k^\downarrow+\frac{1}{d}\geq-\|\eta_j\|\|\eta_k\|+\frac{1}{d}=-d+\frac{2}{d}.
	\end{equation}
\end{proof}

The next theorem is proved in supplementary material. 
\begin{theorem}\label{thm:NegC}
	The channel negativity of any NQPR $\{Q_j\}$ in dimension $d$ satisfies $ N_\rmC(\{Q_j\})\geq N_\rmC^-$.
	The  bound is saturated iff $\{Q_j\}$ is perfect.
\end{theorem}

In view of \thsref{thm:NegBound}, \ref{thm:NegU},  and \ref{thm:NegC}, it is interesting to clarify  states, unitary transformations, and quantum channels with a nonnegative representation. For the Wootters discrete Wigner function in an odd dimension, such pure states and unitary transformations happen to be stabilizer states and Clifford unitaries \cite{Gros06}, which play a key role in quantum computation.  In general, it is not easy to determine  states and quantum channels with a nonnegative representation, but there is a simple description of such unitary transformations.

\begin{theorem}\label{thm:PositiveSym}
A unitary transformation has a nonnegative representation in a  NQPR  $\{Q_j\}$ iff it belongs to the symmetry group of $\{Q_j\}$.
 \end{theorem}
 The symmetry group  can be determined by an algorithm introduced by the author, which is originally designed for computing the symmetry group of a SIC \cite{Zhu12the}.

 \begin{proof}
 If $U$ belongs to the symmetry group of $\{Q_j\}$, then $U^Q$ is a permutation matrix. Conversely, if $U^Q$ has no negative entries, then $U^Q$ is both an orthogonal matrix and a stochastic matrix. So it is a  permutation matrix, and $U$ belongs to the symmetry group.
 \end{proof}

\Thsref{thm:NegBound} to \ref{thm:PositiveSym} elucidate fundamental features of NQPRs. 
\Thsref{thm:NegBound}, \ref{thm:NegU}, and  \ref{thm:NegC} also establish a one-to-one correspondence between  SICs and NQPRs of quantum mechanics with minimal negativity, namely, perfect QPRs. Based on the study of SICs \cite{Zaun11,ReneBSC04, Appl05, ScotG10, ApplFZ15G}, we can construct perfect QPRs at least in  dimensions~2 to~16, 19, 24, 28, 31, 35, 37, 43, 48, and with extremely high precision up to dimension 67 (the number of  known solutions is increasing continually). There is every reason to believe that such representations exist for any Hilbert space of finite dimension. Quite surprisingly, all perfect QPRs  constructed from known SICs  are group covariant and, barring one exception, are covariant with respect to the HW groups, although their definition does not in any way involve group symmetry. This conclusion can be proved rigorously in dimensions 2 (as mentioned before) and 3 according to \thref{thm:QuasiRep3d} below.

\begin{theorem}\label{thm:QuasiRep3d}
All perfect QPRs in dimension 3 are covariant with respect to the HW group. Each one is equivalent to exactly one representation of the form $\{Q_{jk}(t)\}_{j,k=0,1,2}$, with
\begin{equation}\label{eq:QuasiRep3d}
\begin{aligned}
Q_{jk}(t)&=2X^jZ^k |\psi(t)\rangle\langle \psi(t)| (X^{j} Z^{k})^\dag-\frac{1}{3},\\
|\psi(t)\rangle&\rep\frac{1}{\sqrt{2}}(0,1,-\rme^{\rmi t})^\rmT,\quad 0\leq t\leq \frac{\pi}{9},
\end{aligned}
\end{equation}
where the cyclic shift operator $X$  and the phase operator $Z\rep\diag(1,\rme^{2\pi\rmi/3},\rme^{4\pi\rmi/3})$  generate the HW group. 
\end{theorem}
\begin{proof}
According to \rscite{HughS16,Szol14}, all SICs in dimension~3 are covariant with respect to the HW group. According to \rcite{Zhu10} (see also \rcite{Appl05}), every HW covariant SIC in dimension 3 is equivalent to a unique SIC generated from the fiducial state $|\psi(t)\rangle$ with $0\leq t\leq \frac{\pi}{9}$. Now the theorem follows from \thref{thm:NegBound}.
\end{proof}
The SIC generated from the fiducial state $|\psi(0)\rangle$, known as the Hesse SIC \cite{Beng10, TabiA13,  Zhu15S}, deserves special attention. It is the only SIC
that is covariant with respect to the whole Clifford group \cite{Appl05,Zhu10,Zhu15S}. The corresponding NQPR defined in \thref{thm:QuasiRep3d} will be referred to as the Hesse  representation, which stands out  because of its exceptionally high symmetry. In dimension~3, the negativity of the Wootters discrete Wigner function saturates the upper bound in
\thref{thm:NegBound}, and the phase point operators have the form $Q_j^+$, with $\Pi_j$ forming the Hesse SIC. In addition, the Hesse SIC projectors happen to be the states with maximal negativity, which are   useful  in magic state quantum computation \cite{BravK05,VeitMGE14}.

Although all perfect NQPRs known so far are group covariant, their symmetry groups are quite restricted compared with the counterpart of  Wootters discrete Wigner functions.  In odd prime power dimensions, for example, the Wootters discrete Wigner functions are covariant with respect to the Clifford groups (of multipartite HW groups), which form 2-designs \cite{Gros06}. It turns out that the symmetry group of an operator basis  forms a 2-design iff it acts doubly transitively on basis operators, reminiscent of the symplectic geometry in classical phase space \cite{Zhu16P, Spek16}. In sharp contrast, there is only one perfect QPR that is covariant with respect to the Clifford group, namely, the  Hesse representation in dimension 3, according to \thref{thm:NegBound} and early results on SICs \cite{Zhu15S,Zhu16P}. In addition, there are only three
perfect QPRs whose symmetry groups form 2-designs, namely, the representations constructed from the SIC in dimension 2, the Hesse SIC, and the set of Hoggar lines, respectively \cite{Zhu15S,Zhu16P}.
Here we see an interesting tradeoff between negativity and symmetry in QPRs. Wootters discrete Wigner functions are distinguished by extremely high symmetry at the price of high negativity. By contrast, perfect QPRs are distinguished by the minimal negativity at the price of restricted symmetry. According to \thref{thm:PositiveSym}, the latter can only admit restricted unitary transformations with a nonnegative representation.

In summary we introduced the concept of  NQPRs, which include most discrete Wigner functions and   enjoy a number of merits compared with other representations. 
We then elucidate their fundamental features and set the stage for future exploration.
In particular, we   established a one-to-one correspondence between SICs and representations with minimal negativity, namely, perfect QPRs. Surprisingly, most of these representations are automatically covariant with respect to the HW groups. Furthermore, we   revealed an interesting tradeoff between negativity and symmetry in QPRs. 
Our study offers valuable insight on the distinction between quantum theory and classical probability theory in terms of negativity. As a bonus, it
 provides for the first time a quantitative argument for the Born rule representation sought in quantum Bayesianism.
Besides the foundational significance, our work also has ramifications  in practical applications, such as  quantum computation. 
 Finally, our work  prompts several interesting questions, for example, how is quantum theory distinguished from generalized probability theories in terms of negativity? 
The implications of our work for contextuality also deserve further study.

 The author is grateful to Ingemar Bengtsson, Christopher A. Fuchs, and Blake Stacey  for comments and suggestions. The author   acknowledges financial support
 from the Excellence Initiative of the German Federal and State Governments
 (ZUK81) and the DFG.

\bibliographystyle{apsrev4-1}

\bibliography{all_references}

\clearpage
\setcounter{equation}{0}
\setcounter{figure}{0}
\setcounter{table}{0}
\setcounter{theorem}{0}
\setcounter{lemma}{0}
\setcounter{remark}{0}

\makeatletter
\renewcommand{\theequation}{S\arabic{equation}}
\renewcommand{\thefigure}{S\arabic{figure}}
\renewcommand{\thetable}{S\arabic{table}}
\renewcommand{\thetheorem}{S\arabic{theorem}}
\renewcommand{\thelemma}{S\arabic{lemma}}
\renewcommand{\theremark}{S\arabic{remark}}


\onecolumngrid
\begin{center}
	\textbf{\large Supplementary material: Quasiprobability representations of quantum mechanics with minimal negativity}
\end{center}

In this supplementary material, we derive Eq.~(7) and prove  Theorem~4 in the main text, both of which are crucial to understanding the channel negativity of NQPRs. For completeness, we also derive an upper bound for the channel negativity.

\section{Derivation of Eq.~(7)}
Define $Q_{k,1}=(|Q_
k|+Q_k)/2$ and $Q_{k,2}=(|Q_k|-Q_k)/2$, where  $|Q_k|=\sqrt{Q_k^2}$ is the unique positive operator whose square is equal to $Q_k^2$. Then  $Q_{k,1}$ and $Q_{k,2}$ are positive operators with orthogonal supports and $Q_k=Q_{k,1}-Q_{k,2}$. 
\begin{align}\label{eq:NegCs1}
d\Lambda^Q_{jk}&=\tr[Q_j \Lambda(Q_k)]=\tr[Q_j \Lambda(Q_{k,1})]-\tr[Q_j \Lambda(Q_{k,2})]\geq \lambda_{j}^{\min}\tr[\Lambda(Q_{k,1})]-\lambda_{j}^{\max}\tr[\Lambda(Q_{k,2})]\nonumber\\
&=\lambda_{j}^{\min}\tr(Q_{k,1})-\lambda_{j}^{\max}\tr(Q_{k,2})
=\frac{1}{2}\bigl[\lambda_{j}^{\min}\tr(|Q_k|+Q_k)-\lambda_{j}^{\max}\tr(|Q_k|-Q_k)]\nonumber\\
&=\frac{1}{2}\bigl[\lambda_{j}^{\min}(\|\lambda_k\|_1+1)-\lambda_{j}^{\max}(\|\lambda_k\|_1-1)]. 
\end{align}
The lower bound for given $j,k$ can always be saturated by some channel. To see this, let $P_1$ be the projector onto the support of $Q_{k,1}$ and $P_2$ the projector onto the orthogonal complement. Let $\rho_{\max}$ and $\rho_{\min}$ be eigenstates of $Q_j$ with maximal and minimal eigenvalues. Define quantum channel $\Lambda_1$ on the support of $P_1$ and $\Lambda_2$ on the support of $P_2$ as follows, $\Lambda_1(\rho)=\rho_{\min}$ and $\Lambda_2(\rho)=\rho_{\max}$. Then the channel $\Lambda$ defined by the equation $\Lambda(\rho)=\Lambda_1(P_1 \rho P_1)+\Lambda_2(P_2 \rho P_2)$ saturates the lower bound in the above equation.

According to \eref{eq:NegCs1},
\begin{align}
N_\rmC(\{Q_j\})&=-d\min_{\Lambda,j,k} \Lambda^Q_{jk} \leq \frac{1}{2} \max_{j,k}\bigl[\lambda_{j}^{\max}(\|\lambda_k\|_1-1)- \lambda_{j}^{\min}(\|\lambda_k\|_1+1) \bigr].
\end{align} 
The upper bound can be saturated 
since the lower bound in \eref{eq:NegCs1} for any given pair $j,k$ can always be saturated by some channel. So we have 
\begin{align}\label{eq:NegCs3}
N_\rmC(\{Q_j\})&=\frac{1}{2} \max_{j,k}\bigl[\lambda_{j}^{\max}(\|\lambda_k\|_1-1)- \lambda_{j}^{\min}(\|\lambda_k\|_1+1) \bigr]=\frac{1}{2}\max_{j,k}\bigl[\lambda_{j}^{\max}(\|\lambda_k\|_1-1)+ |\lambda_{j}^{\min}|(\|\lambda_k\|_1+1)\bigr].
\end{align}

\section{Proof of Theorem~4}
\begin{proof}According to \eref{eq:NegCs3},
	\begin{align}
	N_\rmC(\{Q_j\})
	&\geq\frac{1}{2}\max_{j}\bigl[\lambda_{j}^{\max}(\|\lambda_j\|_1-1)+ |\lambda_{j}^{\min}|(\|\lambda_j\|_1+1)\bigr]\geq d-\sqrt{d+1}+\frac{2}{d}\sqrt{d+1}-\frac{2}{d}= N_\rmC^-,
	\end{align}
	where the second inequality follows from \lref{lem:NegC} below. If the lower bound is saturated, then $\lambda_j^\downarrow$ are equal to $\vec{v}^\downarrow$ in \eref{eq:v} for all $j$, where  $\vec{v}^\downarrow$ is the vector obtained by arranging the components of $\vec{v}$ in nonincreasing order. 
	So $Q_j$ have the form $Q_j^-$ in Eq.~(12) with $\Pi_j$ rank-1 projectors. The requirement $\tr(Q_jQ_k)=d\delta_{jk}$ then implies that $\tr(\Pi_j\Pi_k)=(d\delta_{jk}+1)/(d+1)$; in other words, $\{\Pi_j\}$ is a SIC and $\{Q_j\}$ is perfect. In this case, we indeed have $N_\rmC(\{Q_j\})=N_\rmC^-$.
\end{proof}

\begin{lemma}\label{lem:NegC}
	Suppose $\vec{v}=(v_1,v_2,\ldots,v_d)$ is a real vector in dimension $d$ satisfying  $\sum_j v_j=1$ and  $\sum_j v_j^2=d$. Let $v_{\max}=\max_j v_j$, $v_{\min}=\min_j v_j$,  $\|v\|_1=\sum_j |v_j|$.
	Then
	\begin{equation}\label{eq:ChannelBound}
	\frac{1}{2}\bigl[v_{\max} (\|v\|_1-1)+|v_{\min}| (\|v\|_1+1)\bigr]\geq d-\sqrt{d+1}+\frac{2}{d}\sqrt{d+1}-\frac{2}{d},
	\end{equation}
	and the lower bound is saturated iff 
	\begin{equation}\label{eq:v}
	v_1^\downarrow=\frac{(d-1)\sqrt{d+1}+1}{d}, \quad
	v_2^\downarrow=\cdots=v_d^\downarrow=\frac{1-\sqrt{d+1}}{d}.
	\end{equation}
\end{lemma}

\begin{proof}The lemma is trivial when $d=2$ since $\vec{v}^\downarrow$ is completely determined by the assumption in the lemma. When  $d\geq3$,
	let us first  consider the special case in which $\vec{v}$ has $m$ components equal to $a>0$, $n$ components equal to $-b<0$, and all other components equal to 0, where $m,n\geq1$ and $m+n\leq d$. By the assumption in the lemma, 
	\begin{equation}
	ma-nb=1,\quad ma^2+nb^2=d,
	\end{equation} 
	which  determines $a$ and $b$ as functions of $m,n$,
	\begin{align}
	a=\frac{1}{m+n}\biggl( 1+\sqrt{\frac{n}{m}(dm+dn-1)}\biggr),\quad
	b=\frac{1}{m+n}\biggl( -1+\sqrt{\frac{m}{n}(dm+dn-1)}\biggr).
	\end{align}	
	Accordingly,
	\begin{align}
	&\frac{1}{2}\bigl[v_{\max} (\|v\|_1-1)+|v_{\min}| (\|v\|_1+1)\bigr]=f(m,n):=\frac{1}{2}\bigl[a(ma+nb-1)+b(ma+nb+1)\bigr]\nonumber\\
	&=\frac{m-n}{m+n}\sqrt{\frac{dm+dn-1}{mn}}-\frac{2}{m+n}+d.
	\end{align}
	Although the vector $\vec{v}$ is defined only when $m,n$ are positive integers. The function $f(m,n)$ is well defined when  $m,n$ are positive real numbers.  When  $d\geq 3$, $m,n\geq 1$, and $m+n\leq d$, calculation shows that $\partial f/\partial m>0 $ and $\partial f/\partial n<0 $. Consequently,
	\begin{equation}
	f(m,n)\geq f(1,d-1)=d-\sqrt{d+1}+\frac{2}{d}\sqrt{d+1}-\frac{2}{d},
	\end{equation}
	and the lower bound is attained only
	when $m=1$ and $n=d-1$, in which case $\vec{v}^\downarrow$ is given by \eref{eq:v}. Therefore, the lemma hods if nonzero components of $\vec{v}$ take on only two distinct values.

	Now we are ready to prove the lemma in general. 
	Without loss of generality, we may assume that $\vec{v}$ has $m+n$ nonzero components, with its first $m$ components being positive and the last $n$ components being negative, where $m,n\geq 1$ and $m+n\leq d$. Define 
	\begin{equation}
	\begin{aligned}
	a=\frac{\sum_{j=1}^m v_j^2}{\sum_{j=1}^m v_j},\quad  m'=\frac{(\sum_{j=1}^m v_j)^2}{\sum_{j=1}^m v_j^2},\quad
	b=\frac{\sum_{j=d-n+1}^d v_j^2}{\sum_{j=d-n+1}^d |v_j|}, \quad n'=\frac{(\sum_{j=d-n+1}^d v_j)^2}{\sum_{j=d-n+1}^d v_j^2}.
	\end{aligned}
	\end{equation}
	Then $a, m', b, n'$ satisfy the following equation,
	\begin{align}
	m'a =\sum_{j=1}^m v_j, \quad      m'a^2=\sum_{j=1}^m v_j^2,\quad 
	-n'b =\sum_{j=d-n+1}^d v_j,\quad    n'b^2=\sum_{j=d-n+1}^d v_j^2.
	\end{align}
	In addition, $0\leq a\leq v_{\max}$,  $0\leq b\leq |v_{\min}|$, $1\leq m'\leq m$, $1\leq n'\leq n$, $m'+n'\leq d$, and $m'a+n'b>1$. Therefore,
	\begin{align}
	&\frac{1}{2}\bigl[v_{\max} (\|v\|_1-1)+|v_{\min}| (\|v\|_1+1)\bigr]=\frac{1}{2}\bigl[v_{\max} (m'a+n'b-1)+|v_{\min}| (m'a+n'b+1)\bigr]\nonumber\\
	& \geq \frac{1}{2}\bigl[a (m'a+n'b-1)+b (m'a+n'b+1)\bigr] =f(m',n')\geq f(1,d-1)=d-\sqrt{d+1}+\frac{2}{d}\sqrt{d+1}-\frac{2}{d}.
	\end{align}
	Here the first inequality is saturated iff $a=v_{\max}$ and $b=|v_{\min}|$, which can happen iff nonzero components of $\vec{v}$ take on only two distinct values. The second inequality is saturated iff $m'=1$ and $n'=d-1$, which imply that $m=1$ and $n=d-1$. If the two inequalities are saturated simultaneously, then $\vec{v}^\downarrow$ necessarily has the form in \eref{eq:v}. 
\end{proof}

\section{Upper bound for the channel negativity}
For completeness, here we derive an upper bound for the channel negativity. According to \eref{eq:NegCs3},
\begin{align}
N_\rmC(\{Q_j\})&=\frac{1}{2} \max_{j,k}\bigl[\lambda_{j}^{\max}(\|\lambda_k\|_1-1)- \lambda_{j}^{\min}(\|\lambda_k\|_1+1) \bigr]
\leq\frac{1}{2}\max_{j}\bigl[\lambda_{j}^{\max}(d-1) -\lambda_{j}^{\min}(d+1)\bigr]\nonumber\\ &\leq\frac{d-1}{\sqrt{2}d}\sqrt{(d+1)(d^2+d+2)}-\frac{1}{d}.
\end{align}
Here  the two inequalities follow from \lsref{lem:norm1} and \ref{lem:NegC} below, respectively. When $d$ is odd, the upper bound can be saturated when one element in $\{Q_j\}$ has the same eigenvalues as a  Wootters phase point operator, that is, $\pm1$ with multiplicity $(d\pm1)/2$, and another element has the eigenvalues determined by \eref{eq:NegC2v}. When $d$ is even, the upper bound cannot be saturated exactly, but can be approached approximately with relative deviation around $1/d^2$.

\begin{lemma}\label{lem:norm1}
	Suppose $\vec{v}=(v_1,v_2,\ldots,v_d)$ is a real vector in dimension $d$ satisfying  $\sum_j v_j=1$ and  $\sum_j v_j^2=d$. Then 
	\begin{equation}
	\|\vec{v}\|_1\leq \begin{cases}
	d & \mbox{$d$ is odd,}\\
	\sqrt{d^2-1} & \mbox{$d$ is even.}
	\end{cases}
	\end{equation}
	When $d$ is odd, the upper bound is saturated iff $\vec{v}$ has $(d+1)/2$ components equal to $1$ and $(d-1)/2$ components equal to $-1$. When $d$ is even, the upper bound is saturated iff $\vec{v}$ has $d/2$ components equal to	$(1\pm\sqrt{d^2-1})/d$, respectively. 
\end{lemma}
\begin{proof}
	When $d$ is odd, the inequality in the lemma is trivial. When the inequality is saturated, all components of $\vec{v}$ must have absolute value 1. The assumption $\sum_j v_j=1$ then implies that $\vec{v}$ has $(d+1)/2$ components equal to $1$ and $(d-1)/2$ components equal to $-1$.
	
	When $d$ is even and $\|\vec{v}\|_1$ is maximized, then all positive components of $\vec{v}$ are equal, and so are all negative components. Now it is straightforward to verify the lemma.
\end{proof}

\begin{lemma}\label{lem:NegC}
	Suppose $\vec{v}=(v_1,v_2,\ldots,v_d)$ is a real vector in dimension $d$ satisfying  $\sum_j v_j=1$ and  $\sum_j v_j^2=d$. Let $v_{\max}=\max_j v_j$, $v_{\min}=\min_j v_j$. Then
	\begin{equation}\label{eq:NegC2}
	\frac{(d-1)v_{\max}-(d+1)v_{\min}}{2}\leq \frac{d-1}{\sqrt{2}d}\sqrt{(d+1)(d^2+d+2)}-\frac{1}{d},
	\end{equation}
	and the upper bound is saturated iff 
	\begin{equation}\label{eq:NegC2v}
	v_1^\downarrow=\frac{1}{d}+\frac{d^2-d+2}{d}\sqrt{\frac{d+1}{2(d^2+d+2)}}, \;
	v_2^\downarrow=\cdots=v_{d-1}^\downarrow=\frac{1}{d}+\frac{2}{d}\sqrt{\frac{d+1}{2(d^2+d+2)}}, \; v_d^\downarrow=\frac{1}{d}-\frac{d^2+d-2}{d}\sqrt{\frac{d+1}{2(d^2+d+2)}}.
	\end{equation}
\end{lemma}
\begin{proof}
	When $d=2$, $v_{\max}=(1+\sqrt{3})/2$ and  $v_{\min}=(1-\sqrt{3})/2$ are determined by the assumption, and the lemma holds. When $d\geq3$ and the left hand side  in \eref{eq:NegC2} is maximized, $\vec{v}^\downarrow$ necessarily has the form $\vec{v}^\downarrow=(a,b,b,\ldots,b,c)$, where $a,b,c$ satisy $a\geq b\geq c$. The assumption in the lemma leads to the two constraints $a+(d-2) b+c=1$ and $a^2+(d-2)b^2+c^2=d$, which determine $c$ as a function of $a$,
	\begin{equation}
	c=\frac{1-a-\sqrt{(d-2)(-da^2+2a+d^2-d-1)}}{d-1}. 
	\end{equation}
	Therefore,
	\begin{align}
	&\frac{(d-1)v_{\max}-(d+1)v_{\min}}{2}=	\frac{(d-1)a-(d+1)c}{2}\nonumber \\
	&=\frac{(d^2-d+2)a+(d+1)\bigl[-1+\sqrt{(d-2)(-da^2+2a+d^2-d-1)}\,\bigr]}{2(d-1)}.
	\end{align}	
	Differentiation with respect to $a$  shows that the maximum is given by \eref{eq:NegC2} and is attained only when
	\begin{equation}
	a=\frac{1}{d}+\frac{d^2-d+2}{d}\sqrt{\frac{d+1}{2(d^2+d+2)}}. 
	\end{equation}
	Consequently,
	\begin{equation}
	b=\frac{1}{d}+\frac{2}{d}\sqrt{\frac{d+1}{2(d^2+d+2)}}, \quad c=\frac{1}{d}-\frac{d^2+d-2}{d}\sqrt{\frac{d+1}{2(d^2+d+2)}},
	\end{equation}
	and $\vec{v}^\downarrow$ has the form in \eref{eq:NegC2v}. 
	
\end{proof}

\end{document}